\documentclass{amsart}
\usepackage[centertags]{amsmath}
\usepackage{amsfonts}
\usepackage{amssymb}
\usepackage{amsthm}
\usepackage{amsmath}
\usepackage{amscd}
\usepackage{latexsym}
\usepackage{color}
\usepackage{graphics}

\theoremstyle{plain}
\newtheorem{thm}{Theorem}[section]
\newtheorem{cor}[thm]{Corollary}

\newtheorem{prop}[thm]{Proposition}
\newtheorem{rem}[thm]{Remark}
\theoremstyle{definition}
\theoremstyle{remark}
\numberwithin{equation}{section}

 \DeclareMathOperator{\IM}{Im}

\begin{document}
\title[Self-dual-anti-self-dual solutions of discrete Yang-Mills equations]{Self-dual and anti-self-dual solutions of discrete Yang-Mills equations  on a double complex}
\author{Volodymyr Sushch}

\address{Department of Mathematics, Koszalin University of Technology, Sniadeckich 2,
 75-453 Koszalin, Poland;  Pidstrygach Institute for Applied Problems of Mechanics and
 Mathematics NASU,  Lviv, Ukraine }

\email{volodymyr.sushch@tu.koszalin.pl}


\date{September 16, 2009}
\maketitle
\begin{abstract}
 We study a  discrete model of the $SU(2)$
 Yang-Mills equations on a combinatorial analog
 of $\Bbb{R}^4$.
 Self-dual and anti-self-dual solutions of discrete Yang-Mills equations are constructed.
 To obtain these solutions we use both techniques of a double complex and the quaternionic approach.
  Interesting analogies between instanton, anti-instanton solutions of discrete and continual self-dual,
  anti-self-dual equations are also discussed.
\end{abstract}

\section{Introduction}

It is well known that the self-dual and anti-self-dual connections
are the absolute minima of the Lagrangian for a $4$-dimensional
non-abelian gauge theory. The first self-dual solution - the one
instanton - to the $SU(2)$ Yang-Mills equations on $\Bbb{R}^4$ was
obtained by Belavin et al \cite{Belavin}. Later other more general
multi-instanton solutions were described in \cite{CF, Jackiw}. Since
then numerous extensions have been made.  Classical references are
the books by Atiyah \cite{Atiyah}, Freed and Uhlenbeck \cite{FU}.

In this paper we study a  discrete analog  of the $SU(2)$
 Yang-Mills equations on a combinatorial analog
 of $\Bbb{R}^4$. The ideas presented here are strongly influenced by book of Dezin \cite{Dezin}.
 We develop discrete models of some objects in differential geometry,
 including the Hodge star operator, the differential and the covariant exterior differential operator,
  in such a way that they preserve the geometric structure of their continual
  analogs.
 We continue the investigations which were originated in
 \cite{Dezin1, S1, S2, S3}.  The purpose of this paper is to construct the
 self-dual and anti-self-dual solutions of discrete $SU(2)$
 Yang-Mills equations which imitate the corresponding solutions of
  continual theory. The geometrical discretisation techniques used here extend
 those introduced in \cite{Dezin} and \cite{S1}. A  combinatorial model of
$\Bbb{R}^4$ based on the use of the double complex construction is
taken from \cite{S3}.

There are many other approaches to the discretisation of Yang-Mills
theories. Numerous papers have been written on this subject. See,
for example, \cite{SS, Cast, FJ, GM,  KN, N, Nish, Shab, Se} and the references therein.
Most of them are based on the lattice discretisation scheme.
However, in the case of the lattice formulation there are
difficulties in keeping geometrical properties of an origin gauge
theory. An alternative geometrical discretisation scheme of a field
theory can be found in \cite{SSSA}.

 The paper is organized as follows. In Section 2 we review
some basic facts of the $SU(2)$ Yang-Mills theory on $\Bbb{R}^4$. We
begin by recalling the connection between the  Lie group $SU(2)$ and
the space of quaternions. Finally, we write down  the basic
instanton and anti-instanton solutions in quaternionic form. The
notations here are compiled from \cite{Atiyah} and \cite{NS}.

Section 3 contains a brief summary of definitions and properties due
to the double complex construction. We repeat here the relevant
material from \cite{S3}. This article is also the main reference for
this section. In particular, we introduce  discrete matrix-valued
forms (analog of differential forms) and define analogs of the main
continual operations on them.

In Section 4 using the quaternionic approach we present the discrete
Yang-Mills equations. We write out components of the discrete
curvature 2-form in quaternionic form. The discrete self-dual and
anti-self-dual equations are described.  We try to be as close to
continual $SU(2)$ Yang-Mills theory as possible.  Hence we discuss
conditions when the discrete curvature will be $su(2)$-valued.

Finally, Section 5 is devoted to self-dual and anti-self-dual
solutions of the discrete Yang-Mills equations. We construct these
solutions as discrete quaternionic 1-forms and discuss some
analogies with continual instanton and anti-instanton solutions.

\section{Quaternions and $SU(2)$-connection}

In this section we briefly recall some well known settings of the
smooth Yang-Mills theory in Euclidean 4-dimensional space (see, for
example, \cite{NS}).

We  begin with a brief review of some preliminaries about
quaternions.   The quaternions are formed from real numbers by
adjoining three symbols $\mathbf{i}, \mathbf{j}, \mathbf{k}$ and an
arbitrary quaternion $x$ can be written as
\begin{equation}\label{2.1}
x=x_{1}+x_{2}\mathbf{i}+x_{3}\mathbf{j}+x_{4}\mathbf{k},
\end{equation} where
$x_{1}, x_{2}, x_{3}, x_{4}\in \mathbb{R}$. The  symbols
$\mathbf{i}, \mathbf{j}, \mathbf{k}$ satisfy the following
identities
\begin{align}\label{2.2}
&\mathbf{i}^2=\mathbf{j}^2=\mathbf{k}^2=-1,\notag \\
&\mathbf{i}\mathbf{j}=-\mathbf{j}\mathbf{i}=\mathbf{k}, \quad
\mathbf{j}\mathbf{k}=-\mathbf{k}\mathbf{j}=\mathbf{i}, \quad
\mathbf{k}\mathbf{i}=-\mathbf{i}\mathbf{k}=\mathbf{j}.
\end{align}
It is clear that the space of quaternions is isomorphic to
$\mathbb{R}^4$. By analogy with the complex numbers $x_{1}$ is
called the real part of $x$ and
$x_{2}\mathbf{i}+x_{3}\mathbf{j}+x_{4}\mathbf{k}$ is called the
imaginary part. In further we will write
$$\IM{x}=x_{2}\mathbf{i}+x_{3}\mathbf{j}+x_{4}\mathbf{k}.$$
The conjugate quaternion of $x$ is defined by
$$\bar{x}=x_{1}-x_{2}\mathbf{i}-x_{3}\mathbf{j}-x_{4}\mathbf{k}.$$
Then  the norm $|x|$ of a quaternion can be introduced as follows
\begin{equation}\label{2.3}
|x|^2=x\bar{x}=x_{1}^2+x_{2}^2+x_{3}^2+x_{4}^2.
\end{equation}
If $x\neq0$, then it has a unique inverse $x^{-1}$ given by
\begin{equation}\label{2.4}
x^{-1}=\bar{x}/|x|^2.
\end{equation}
The algebra of quaternions can be represented as a sub-algebra of
the $2\times2$ complex matrices $M(2, \mathbb{C})$. We identify the
quaternion (\ref{2.1}) with a matrix $f(x)\in M(2, \mathbb{C})$ by
setting
\begin{equation}\label{2.5}
f(x)=\left(
       \begin{array}{cc}
         x_1+x_2i & x_3+x_4i \\
         -x_3+x_4i & x_1-x_2i \\
       \end{array}
     \right).
\end{equation}
Here $i$ is the imaginary unit.

 It is well known that the unit
quaternions, i.e., they have norm $|x|=1$, form a group and this group
is isomorphic to $SU(2)$.  The following  $2\times2$  complex
matrices
\begin{equation}\label{2.6}
\mathbf{i}=\left(
              \begin{array}{cc}
                i & 0 \\
                0 & -i \\
              \end{array}
            \right), \quad \mathbf{j}=\left(
                             \begin{array}{cc}
                               0 & 1 \\
                               -1 & 0 \\
                             \end{array}
                           \right), \quad \mathbf{k}=\left(
                                                       \begin{array}{cc}
                                                         0 & i \\
                                                         i & 0 \\
                                                       \end{array}
                                                     \right)
\end{equation}
realize a representation of the Lie algebra $su(2)$ of the  group
$SU(2)$. Note that multiplying by $-i$ these tree matrices we obtain
the standard Pauli matrices. Matrices (\ref{2.6}) correspond to the
units $\mathbf{i}, \mathbf{j}, \mathbf{k}$ given by (\ref{2.2}).
Thus the Lie algebra $su(2)$ can be viewed as the pure imaginary
quaternions with basis $\mathbf{i}, \mathbf{j}, \mathbf{k}$.

Let now $A$ be an $SU(2)$-connection. This means that $A$ is an
$su(2)$-valued 1-form and we can write
\begin{equation}\label{2.7}
A=\sum_{\mu}A_{\mu}(x) dx^\mu,
\end{equation}
where $A_{\mu}(x)\in su(2)$ and  $x=(x_1, ...,  x_4)$ is a point of
$\mathbb{R}^4$. The connection $A$  is  also called a gauge
potential. Define a gauge transformation by a function $g(x)$ taking
value in $SU(2)$. Then the gauge potential $A$ must transform like
\begin{equation} \label{2.8}
A \rightarrow  g^{-1}Ag+g^{-1}dg.
 \end{equation}

Let us define the curvature 2-form $F$ by
\begin{equation} \label{2.9}
 F=dA+A\wedge A,
 \end{equation}
 where $\wedge$ denotes the exterior multiplication.

Consider also  the covariant exterior differential operator $d_A$ given by
\begin{equation} \label{2.10}
d_A\Omega=d\Omega+A\wedge \Omega+(-1)^{p+1}\Omega\wedge A,
\end{equation} where $\Omega$ is a $su(2)$-valued $p$-form.

The Yang-Mills action S can be
expressed in terms of the 2-forms $F$ and $\ast F$ as
\begin{equation} \label{2.11}
 S=-tr\int_{\mathbb{R}^{4}}F\wedge\ast F,
 \end{equation}
 where $\ast$ is the Hodge star operator. In $\mathbb{R}^{4}$ the operator $\ast^2$ is either an involution or
anti-involution, i.e.,  $\ast^2=\pm1$.
 The Yang-Mills Lagrangian $L=-tr(F\wedge\ast F)$ is invariant under the
 gauge transformation (\ref{2.8}).
 By the physical requirement it is clear that the action $S$ should be finite. Hence the curvature $F$ should be square integrable.
 This means that $F\rightarrow 0$ as $|x|\rightarrow \infty$. Consequently, we must describe the boundary condition at infinity for the connection $A$. By virtue of gauge freedom  (\ref{2.8}) we have
 \begin{equation} \label{2.12}
A \sim  g^{-1}dg \qquad  \mbox{as} \qquad  |x|\rightarrow \infty,
 \end{equation}
 where $\sim$ implies asymptotic behaviour. Here and subsequently we do not specify the rate of decay.

 Written in terms of the covariant exterior differential operator $d_A$
 the Euler-Lagrange  equations for the extrema of (\ref{2.11}) have the form
\begin{equation} \label{2.13}
d_A F=0, \qquad d_A\ast F=0.
\end{equation}
These equations are the Yang-Mills equations.  The first equation of (\ref{2.13}) is known also as the Bianchi identity.
In $4$-dimensional Yang-Mills
 theories the following equations
 \begin{equation} \label{2.14}
 F=\ast F, \qquad F=-\ast F
 \end{equation}
 are called self-dual and anti-self-dual respectively.
 These equations  are first-order non-linear equations for the potential $A$ which
imply the second-order Yang-Mills equations (\ref{2.13}).
 Solutions of (\ref{2.14}) -- the self-dual and anti-self-dual
 connections -- are called also instantons and anti-instantons \cite{FU}.
 It is known that the self-dual and anti-self-dual connections are the
 absolute minima of the action $S$.

The connection 1-form $A$ can be
defined also  as taking values in the space of pure imaginary
quaternions. To express $A$ in quaternion form we consider the quaternion differential
$$dx=dx_{1}+dx_{2}\mathbf{i}+dx_{3}\mathbf{j}+dx_{4}\mathbf{k}$$
and the conjugate quaternion of $dx$
$$d\bar{x}=dx_{1}-dx_{2}\mathbf{i}-dx_{3}\mathbf{j}-dx_{4}\mathbf{k}.$$
Let $f(x)$ be a function of the quaternion variable $x$   with quaternion values. Then we can write $A$ as
\begin{equation}\label{2.16}
A=\IM(f(x)dx),
\end{equation}
where
\begin{equation*}
f(x)=f_{1}(x)+f_{2}(x)\mathbf{i}+f_{3}(x)\mathbf{j}+f_{4}(x)\mathbf{k}.
\end{equation*}
 Using the rules of multiplication (\ref{2.2}) we have
 \begin{align*}
 &A_1(x)=f_{2}(x)\mathbf{i}+f_{3}(x)\mathbf{j}+f_{4}(x)\mathbf{k},\\
 &A_2(x)=f_{1}(x)\mathbf{i}+f_{4}(x)\mathbf{j}-f_{3}(x)\mathbf{k},\\
 &A_3(x)=-f_{4}(x)\mathbf{i}+f_{1}(x)\mathbf{j}+f_{2}(x)\mathbf{k},\\
 &A_4(x)=f_{3}(x)\mathbf{i}-f_{2}(x)\mathbf{j}+f_{1}(x)\mathbf{k}.
 \end{align*}
Using (\ref{2.16}) we can rewrite (\ref{2.9}) as follows
 \begin{equation}\label{2.17}
F=\IM(df(x)\wedge dx+f(x)dx\wedge f(x)dx).
\end{equation}
Note that  calculation of the imaginary part of $f(x)dx$ and
computing its curvature commute.

 Let us take the following expression for $f(x)$:
 \begin{equation}\label{2.18}
 f(x)=\frac{\bar{x}}{1+|x|^{2}}.
 \end{equation}
 Then the connection 1-form $A$ is  defined by
 \begin{equation}\label{2.19}
A=\IM\Big\{\frac{\bar{x}dx}{1+|x|^{2}}\Big\}.
\end{equation}
The explicit components $A_{\mu}$ can be written as
\begin{align}\label{2.20}
A_1(x)=\frac{-x_{2}\mathbf{i}-x_{3}\mathbf{j}-x_{4}\mathbf{k}}{1+|x|^{2}},  \qquad
A_2(x)=\frac{x_{1}\mathbf{i}-x_{4}\mathbf{j}+x_{3}\mathbf{k}}{1+|x|^{2}}, \notag \\
A_3(x)=\frac{x_{4}\mathbf{i}+x_{1}\mathbf{j}-x_{2}\mathbf{k}}{1+|x|^{2}}, \qquad
A_4(x)=\frac{-x_{3}\mathbf{i}+x_{2}\mathbf{j}+x_{1}\mathbf{k}}{1+|x|^{2}}.
\end{align}
 Putting (\ref{2.18}) in (\ref{2.17}) we get the pure imaginary expression
 \begin{equation}\label{2.21}
F=\frac{d\bar{x}\wedge dx}{(1+|x|^{2})^2}.
\end{equation}
It is easy  to show that the 2-form $d\bar{x}\wedge dx$ is
anti-self-dual. Hence $F$ is anti-self-dual too and the connection
(\ref{2.19}) describes an anti-instanton . See for details
\cite{Atiyah}.

Similarly, if we take
\begin{equation}\label{2.22}
A=\IM\Big\{\frac{xd\bar{x}}{1+|x|^{2}}\Big\},
\end{equation}
then we obtain the self-dual 2-form
\begin{equation}\label{2.23}
F=\frac{dx\wedge d\overline{x}}{(1+|x|^{2})^2}.
\end{equation}
Thus the curvature is self-dual  and (\ref{2.22}) describes an
instanton .

 \section {Double complex}
 We will need the double complex construction described in \cite{S3}.
 In with section for the convenience of the  reader we repeat the relevant material
  from \cite{S3} without proofs, thus making our presentation  self-contained.

 Let the tensor product  $C(4)=C\otimes C\otimes C\otimes C$
of an  1-dimensional complex $C$ be a combinatorial model of
Euclidean space $\Bbb{R}^4$ (see for details also \cite{Dezin}).
The 1-dimensional complex $C$ is defined in the following way. Let
$C^0$ denotes the real linear space of 0-dimensional chains
generated by basis elements $x_j$ (points), $j\in
\Bbb{Z}$. It is convenient to introduce the shift operators
$\tau,\sigma$ in the set of indices by
\begin{equation} \label{3.1} \tau j=j+1,
\qquad \sigma j=j-1.
\end{equation}
We denote the open interval $(x_j, \ x_{\tau j})$ by $e_j$. We'll
regards the set $\{e_j\}$ as a set of basis elements of the real
linear space $C^1$  of 1-dimensional chains. Then the 1-dimensional
complex (combinatorial real line) is the direct sum of the
introduced spaces $C=C^0\oplus C^1$. The boundary operator
$\partial$ on the basis elements of $C$ is given by
 \begin{equation} \label{3.2}
 \partial x_j=0, \qquad  \partial
e_j=x_{\tau j}-x_j.
 \end{equation}
 The definition is extended to arbitrary chains by linearity.

Multiplying the basis elements $x_j, e_j$ in various ways we obtain
basis elements of $C(4)$. Let $s_k^{(p)}$, where
$k=(k_1,k_2,k_3,k_4)$ and  $k_i\in\Bbb Z,$ be an arbitrary basis
element of $C(4)$.  Then a $p$-dimensional chain is given by
\begin{equation}\label{3.3}
c_p=\sum_k\sum_p c^k_{(p)}s_k^{(p)}, \quad c^k_{(p)}\in\Bbb R.
\end{equation}
 We suppose that the superscript $(p)$ contains
the whole requisite information about the quantity and places of
1-dimensional elements  $e_j$  in $s_k^{(p)}$.
 For example, the 1-dimensional basis elements $e_k^i$
of $C(4)$ can be written as
\begin{align}\label{3.4}\notag
 e_k^1&= e_{k_1}\otimes
x_{k_2}\otimes x_{k_3}\otimes x_{k_4}, \qquad e_k^2=
x_{k_1}\otimes e_{k_2}\otimes x_{k_3}\otimes x_{k_4}, \\ e_k^3&=
x_{k_1}\otimes x_{k_2}\otimes e_{k_3}\otimes x_{k_4}, \qquad
e_k^4=x_{k_1}\otimes x_{k_2}\otimes x_{k_3}\otimes e_{k_4}
\end{align}
 and for the 2-dimensional basis elements
$\varepsilon_k^{ij}$ we have
\begin{align}\label{3.5}\notag
\varepsilon_k^{12}&=e_{k_1}\otimes e_{k_2}\otimes x_{k_3}\otimes
e_{k_4}, \qquad \varepsilon_k^{23}=x_{k_1}\otimes e_{k_2}\otimes
e_{k_3}\otimes x_{k_4}, \\ \notag
\varepsilon_k^{13}&=e_{k_1}\otimes x_{k_2}\otimes e_{k_3}\otimes
x_{k_4}, \qquad \varepsilon_k^{24}=x_{k_1}\otimes e_{k_2}\otimes
x_{k_3}\otimes e_{k_4},
 \\ \varepsilon_k^{14}&=e_{k_1}\otimes x_{k_2}\otimes
x_{k_3}\otimes e_{k_4}, \qquad \varepsilon_k^{34}=x_{k_1}\otimes
x_{k_2}\otimes e_{k_3}\otimes e_{k_4}.
\end{align}
Using (\ref{3.2}) we define the boundary operator
$\partial$ on chains of $C(4)$ in the following way: if $c_p, \
c_q$ are chains of the indicated dimension, belonging to the
complexes being multiplied, then
\begin{equation}\label{3.6}
\partial(c_p\otimes c_q)=\partial c_p\otimes c_q+(-1)^pc_p\otimes\partial c_q.
\end{equation}
For example, for the basis element $\varepsilon_k^{24}$ we have
\begin{align*}
\partial\varepsilon_k^{24}&=\partial(x_{k_1}\otimes e_{k_2})\otimes
x_{k_3}\otimes e_{k_4}-x_{k_1}\otimes e_{k_2}\otimes
\partial(x_{k_3}\otimes e_{k_4}) \\ &=\partial x_{k_1}\otimes
e_{k_2}\otimes x_{k_3}\otimes e_{k_4}+x_{k_1}\otimes \partial
e_{k_2}\otimes x_{k_3}\otimes e_{k_4}\\ &- x_{k_1}\otimes
e_{k_2}\otimes \partial x_{k_3}\otimes e_{k_4}-x_{k_1}\otimes
e_{k_2}\otimes x_{k_3}\otimes \partial e_{k_4} \\ &=
x_{k_1}\otimes  x_{\tau k_2}\otimes x_{k_3}\otimes
e_{k_4}-x_{k_1}\otimes  x_{k_2}\otimes x_{k_3}\otimes e_{k_4}\\
&-x_{k_1}\otimes x_{k_2}\otimes x_{k_3}\otimes x_{\tau
k_4}+x_{k_1}\otimes x_{k_2}\otimes x_{k_3}\otimes x_{k_4}.
\end{align*}

For convenience we also  introduce the shift operators $\tau_i$ and
$\sigma_i$ which act in the set of indices $k=(k_1,k_2,k_3,k_4), \
k_i\in\Bbb Z,$ as
\begin{equation}\label{3.7}
\tau_ik=(k_1,...\tau
 k_i,...k_4), \qquad
 \sigma_ik=(k_1,...\sigma k_i,...k_4),
\end{equation}
where $\tau$ and $\sigma$ are given by (\ref{3.1}).

Let us introduce the construction of a double complex. Together with
the complex $C(4)$ we consider its double, namely the complex
$\tilde{C}(4)$ of exactly the same structure. Define the one-to-one
correspondence
\begin{equation}\label{3.8}
\ast : C(4)\rightarrow\tilde{C}(4), \qquad \ast : \tilde
C(4)\rightarrow C(4)
\end{equation}
in the following way. Let  $s_k^{(p)}$ be an arbitrary
$p$-dimensional basis element of $C(4)$, i.e.,  the product
$s_k^{(p)}=s_{k_1}\otimes s_{k_2}\otimes s_{k_3}\otimes s_{k_4}$
contains exactly $p$ of $1$-dimensional elements $e_{k_i}$ and $4-p$
of $0$-dimensional elements  $x_{k_i}$, $p=0,1,2,3,4$, \ $k_i\in\Bbb
Z.$ Then
\begin{equation}\label{3.9}
\ast : s_k^{(p)}\rightarrow\pm\tilde s_k^{(4-p)}, \qquad \ast :
\tilde s_k^{(4-p)}\rightarrow \pm s_k^{(p)},
\end{equation}
where
\begin{equation*}
 \tilde s_k^{(4-p)}=*s_{k_1}\otimes *s_{k_2}\otimes
*s_{k_3}\otimes *s_{k_4}
\end{equation*}
and $*s_{k_i}=\tilde e_{k_i}$ if $s_{k_i}=x_{k_i}$ and
$*s_{k_i}=\tilde x_{k_i}$ if $s_{k_i}=e_{k_i}.$ In the first of
mapping  (\ref{3.9}) we take "$+$" if the permutation $((p), \
(4-p))$ of $(1,2,3,4)$ is even and "$-$" if the permutation $((p), \
(4-p))$ is odd. Recall that in symbol $(p)$ the number of basis
element is contained. For example, for the 2-dimensional basis
element $\varepsilon_k^{13}=e_{k_1}\otimes x_{k_2}\otimes
e_{k_3}\otimes x_{k_4}$ we have
 $\ast\varepsilon_k^{13}=-\tilde\varepsilon_k^{24}$ since the
 permutation $(1,3,2,4)$ is odd. The mapping $\ast :
\tilde s_k^{(4-p)}\rightarrow \pm s_k^{(p)}$ is defined by
analogy.
\begin{prop}Let $c_r\in C(4)$ be an $r$-dimensional chain (\ref{3.3}).
Then we have
\begin{equation}\label{3.10}
\ast\ast c_r=(-1)^{r(4-r)}c_r.
\end{equation}
\end{prop}
\begin{proof} See \cite{S3}.
\end{proof}

Now we  consider a dual object of the  complex  $C(4)$. Let $K(4)$
be a cochain complex with  $gl(2,\Bbb{C})$-valued coefficients,
where  $gl(2,\Bbb{C})$ is the Lie algebra of the group
$GL(2,\Bbb{C})$. Recall that  $gl(2,\Bbb{C})$ consists of all
complex $2\times2$ matrices $M(2,\Bbb{C})$ with bracket operation
$[\cdot , \cdot]$. We suppose that the complex $K(4)$, which is a
conjugate of $C(4)$, has a similar structure: ${K(4)=K\otimes
K\otimes K\otimes K}$, where $K$ is a conjugate of the 1-dimensional
complex $C$. Basis elements of $K$ can be written as ${x^j}, \
{e^j}$. Then an arbitrary basis element of $K(4)$ is given by ${s^k_{(p)}=
s^{k_1}\otimes s^{k_2}\otimes s^{k_3}\otimes s^{k_4}}$, where
$s^{k_j}$ is either $x^{k_j}$ or $e^{k_j}$. For example, we denote
the 1-, 2-dimensional basis elements of $K(4)$ by $e_i^k$, \
$\varepsilon^k_{ij}$ respectively, cf. (\ref{3.4}), (\ref{3.5}). For
a $p$-dimensional cochain $\varphi\in K(4)$ we have
\begin{equation}\label{3.11}
\varphi= \sum_k \sum_p\varphi_k^{(p)}s_{(p)}^k,
 \end{equation}
where $\varphi_k^{(p)}\in gl(2,\Bbb{C})$.
We will
call cochains forms, emphasizing their relationship with the
corresponding continual objects, differential forms.

 We define the pairing operation $<\cdot \ , \ \cdot>$ for arbitrary
basis elements \ $\varepsilon_k\in C(4)$, \ $s^k\in K(4)$ by the
rule
\begin{equation}\label{3.12}
 <\varepsilon_k, as^k>=\left\{\begin{array}{l}0,\
\varepsilon_k\ne s_k\\
                            a,\ \varepsilon_k=s_k,\ a\in gl(2,\Bbb{C}).
                            \end{array}\right.
\end{equation}
Here for simplicity  the superscript $(p)$ is omitted.
 The operation (\ref{3.12}) is linearly extended to cochains.

The operation $\partial$ (\ref{3.6}) induces the dual operation
$d^c$ on $K(4)$ in the following way:
\begin{equation}\label{3.13}
<\partial\varepsilon_k, as^k>=<\varepsilon_k, ad^cs^k>.
\end{equation}
For example, if $\varphi= \sum_k \varphi_kx^k$, where $x^k=x^{k_1}\otimes x^{k_2}\otimes x^{k_3}\otimes x^{k_4}$, is a 0-form, then
\begin{equation}\label{3.14}
d^c\varphi= \sum_k \sum_{i=1}^4(\Delta_i\varphi_k)e_i^k,
\end{equation}
where $\Delta_i\varphi_k=\varphi_{\tau_ik}-\varphi_k$ and $e_i^k$ is the 1-dimensional basis elements of $K(4)$.
 The coboundary operator $d^c$ is an analog of the exterior
differentiation operator.

 Now we describe a cochain product on the forms of $K(4)$. See \cite{Dezin}
 for details. We denote this product by
 $\cup$. In terms of the homology theory this is the so-called Whitney product.
  First we introduce the $\cup$-product on the chains of the 1-dimensional complex K.
  For the basis elements of $K$ the
$\cup$-product is defined as follows $$ x^j\cup
x^j=x^j, \quad e^j\cup x^{\tau j}=e^j,
\quad x^j\cup e^j=e^j, \quad j\in\Bbb{Z}, $$
supposing the product to be zero in all other case. To arbitrary
forms the $\cup$-product be extended linearly. Let us introduce an
$r$-dimensional complex $K(r)$, ${r=1,2,3}$,\ in an obvious
notation. Let $s_{(p)}^k$ be an arbitrary $p$-dimensional basis
element of $K(r)$. It is convenient to write the basis  element of
$K(r+1)$ in the form   $s_{(p)}^k\otimes s^j$, where
$s_{(p)}^k$ is a basis element of $K(r)$ and $s^j$ is either
$e^j$  or $x^j$, \ $j\in\Bbb{Z}$. Then, supposing
that the $\cup$-product in $K(r)$ has been defined, we introduce
it for basis elements of $K(r+1)$ by the rule
\begin{equation}\label{3.15}
(s^k_{(p)}\otimes
s^j)\cup(s^k_{(q)}\otimes s^\mu)= Q(j,q)(s^k_{(p)}\cup
s^k_{(q)})\otimes(s^j\cup s^\mu),
\end{equation}
 where  the signum
function $Q(j, q)$ is equal to $-1$ if the dimension of both
elements $s^j$, $s_{(q)}^k$ is odd and to $+1$ otherwise.  The
extension of the $\cup$-product to arbitrary forms of $K(r+1)$ is
linear.  Note that the coefficients of forms multiply as matrices.

\begin{prop}{Let $\varphi$ and $\psi$ be arbitrary forms of $K(4)$.
Then
\begin{equation}\label{3.16}
 d^c(\varphi\cup\psi)=d^c\varphi\cup\psi+(-1)^p\varphi\cup
d^c\psi,
\end{equation} where  $p$ \ is the dimension of a form
$\varphi$.}
\end{prop}

The proof of Proposition~3.2 is totally analogous to one in
\cite[p.~147]{Dezin} for the case of discrete forms with real
coefficients.

The complex of the cochains $\tilde K(4)$ over the double complex
$\tilde C(4)$ with the operator $d^c$ defined in it by (\ref{3.13})
has the same structure as  $K(4)$. The operation (\ref{3.8}) induces
the respective mapping
\begin{equation*}
\ast : K(4)\rightarrow\tilde{K}(4), \qquad \ast : \tilde
K(4)\rightarrow K(4)
\end{equation*}
by the rule:
\begin{equation}\label{3.17}
<\tilde
c, \ *\varphi>=<*\tilde c, \ \varphi>, \qquad <c, \ *\tilde
\psi>=<*c, \ \tilde\psi>,
\end{equation}
where $c\in C(4), \ \tilde c\in\tilde C(4), \ \varphi\in K(4), \
\tilde\psi\in \tilde K(4)$. Hence for the basic elements of  $K(4)$ or $\tilde K(4)$ we have  relations (\ref{3.9}). It is obviously that Proposition~3.1
is true for any $r$-dimensional cochain $c^r\in K(4)$. So we have
$$\ast\ast\varphi=(-1)^{r(4-r)}\varphi$$ for any discrete $r$-form
$\varphi$ on $K(4)$ and note that the same relation holds for the
Hodge star operator. Thus this operator is a combinatorial analog of
the Hodge star operator.

Let us introduce the following operation  $$\tilde\iota: K(4)
\rightarrow \tilde K(4), \qquad \tilde\iota: \tilde K(4)
\rightarrow K(4)$$ by setting
\begin{equation}\label{3.18}
 \tilde\iota s_{(p)}^k= \tilde s_{(p)}^k, \qquad \tilde\iota\tilde s_{(p)}^k=  s_{(p)}^k,
\end{equation}
where $s_{(p)}^k$ and $\tilde s_{(p)}^k$ are  basis elements of
$K(4)$ and $\tilde K(4) $.  Hence for a $p$-form  $\varphi\in K(4)$
we have \ $\tilde\iota\varphi=\tilde\varphi$. \ Recall that the
coefficients of $\tilde\varphi\in \tilde K(4)$ and  $\varphi\in
K(4)$ are the same.
\begin{prop} The following hold
\begin{align}\label{3.19}
\tilde\iota^2=Id, \quad \tilde\iota\ast&=\ast\tilde\iota, \quad
\tilde\iota d^c=d^c\tilde\iota,\\ \notag
\tilde\iota(\varphi\cup\psi)&=\tilde\iota\varphi\cup\tilde\iota\psi,
\end{align}
where $\varphi, \ \psi\in K(4)$.
\end{prop}
\begin{prop} Let $h$ be a discrete 0-form. Then for an arbitrary $p$-form $\varphi\in
K(4)$ we have
\begin{equation}\label{3.20}
 \tilde\iota\ast(h\cup \varphi)=h\cup\tilde\iota\ast\varphi.
 \end{equation}
\end{prop}

\begin{proof} See \cite{S3}.
\end{proof}

Note that  the definition of  inner product in the double complex
and  a discrete analog of the Yang-Millls actions (\ref{2.11}) can
be found in \cite{S3}.

\section{Quaternions and discrete forms}
Let us consider a discrete 0-form with coefficients belonging to
$M(2, \mathbb{C})$. We put
\begin{equation}\label{4.1}
f=\sum_k f_k x^k,
\end{equation}
where  $x^k=x^{k_1}\otimes x^{k_2}\otimes x^{k_3}\otimes x^{k_4}$ is
the 0-dimensional basis element of $K(4)$, \ $k=(k_1,k_2,k_3,k_4), \
k_i\in\Bbb Z.$ Suppose that the matrices $f_k\in M(2, \mathbb{C})$
look like  (\ref{2.5}), i. e.
\begin{equation}\label{4.2}
f_k=\left(
      \begin{array}{cc}
         f_k^1+f_k^2i & f_k^3+f_k^4i \\
         -f_k^3+f_k^4i & f_k^1-f_k^2i \\
       \end{array}
     \right),
\end{equation}
where $f_k^s\in\mathbb{R}, \ s=1,2,3,4$. Then $f_k$ in quaternionic
form can be expressed as
\begin{equation}\label{4.3}
f_k=f_k^{1}+f_k^{2}\mathbf{i}+f_k^{3}\mathbf{j}+f_k^{4}\mathbf{k}.
\end{equation}
Hence the form  (\ref{4.1}) can be considered as a discrete form with
quaternionic coefficients. We will call it simply the quaternionic form when no confusion can arise.
In a proper way we define the quaternionic
0-form $\bar{f}$ with   coefficients $\bar{f_k}$  regarded as the conjugate quaternions of $f_k$.
Let  $f^{-1}$ be the quaternionic form, where
$f_k^{-1}$ is given by (\ref{2.4}). Then we have
\begin{equation}\label{4.4}
f\cup f^{-1}=\sum_k f_k f_k^{-1}x^k=\sum_k x^k.
\end{equation}
\begin{prop} Let $f$ be a discrete 0-form and $f\neq 0$. Then we have
\begin{equation}\label{4.5}
d^cf\cup f^{-1}=-f\cup d^cf^{-1}.
\end{equation}
\end{prop}
\begin{proof}
By definition  (\ref{3.14}) and according to (\ref{4.4}), we have
 $d^c(f\cup f^{-1})=0$. Using Proposition~3.2 we immediately obtain
 (\ref{4.5}).
\end{proof}

Let us denote by $e$ the following quaternionic 1-form
\begin{equation}\label{4.6}
e=\sum_ke^k=\sum_k (e^k_{1}+e^k_{2}\mathbf{i}+e^k_{3}\mathbf{j}+e^k_{4}\mathbf{k}),
\end{equation}
where $e^k_i$ is the 1-dimensional basis elements of $K(4)$. Let
$A\in K(4)$ be a discrete 1-form. We define the discrete
$SU(2)$-connection $A$ to be
\begin{equation}\label{4.7}
A=\sum_k\sum_{i=1}^4A_k^ie_i^k,
\end{equation}
where \ $A_k^i\in su(2)$ \ and  \ $k=(k_1,k_2,k_3,k_4), \ k_i\in\Bbb
Z.$  Using (\ref{4.3}) and (\ref{4.6}) we write (\ref{4.7}) in quaternionic form as
\begin{equation}\label{4.8}
A=\IM(f\cup e)=\IM\Big(\sum_k f_ke^k\Big).
\end{equation}
Then the $A_k^i$ are given by
\begin{align}\label{4.9}
 &A_k^1=f_k^{2}\mathbf{i}+f_k^{3}\mathbf{j}+f_k^{4}\mathbf{k},\qquad  \ \
 A_k^2=f_k^{1}\mathbf{i}+f_k^{4}\mathbf{j}-f_k^{3}\mathbf{k},\notag\\
 &A_k^3=-f_k^{4}\mathbf{i}+f_k^{1}\mathbf{j}+f_k^{2}\mathbf{k},\qquad
 A_k^4=f_k^{3}\mathbf{i}-f_k^{2}\mathbf{j}+f_k^{1}\mathbf{k}.
 \end{align}
Define the quaternionic 0-form $x$ by
\begin{equation}\label{4.10}
x=\sum_k\kappa x^k, \quad   \kappa=k_{1}+k_{2}\mathbf{i}+k_{3}\mathbf{j}+k_{4}\mathbf{k},
\end{equation}
where $k_i\in\Bbb Z.$
It is easy to check that
\begin{equation}\label{4.11}
d^cx=e.
\end{equation}
Therefore we can rewrite (\ref{4.8}) as
\begin{equation}\label{4.12}
A=\IM(f\cup d^cx).
\end{equation}

Let $g$ be a  quaternionic  0-form (\ref{4.1}) with the components of unit norm, i.e., $|g_k|=1$ for any $k$.
It means that the corresponding discrete form is $SU(2)$-valued.
We  now define a gauge transformation for the discrete potential $A$ which is analogous to  (\ref{2.8}).
This is
\begin{equation}\label{4.13}
A \rightarrow  g^{-1}\cup A\cup g+g^{-1}\cup d^cg,
 \end{equation}
where $A$ is given by (\ref{4.8}) or (\ref{4.12}).
Note that the gauge transformed discrete form $A$ is $su(2)$-valued too.  It is not so obviously as in the continual case
but follows immediately from the definition of $\cup$-multiplication and formula (\ref{3.16}).
More generally, if we assume that the gauge
transformation $g$ is an arbitrary quaternionic 0-form, then we  take
the imaginary part of $g^{-1}\cup A\cup g+g^{-1}\cup d^cg$ in (\ref{4.13}).   For a deeper discussion of
 gauge invariant discrete models of the Yang-Mills theory we refer the reader to
\cite{S1, S3}.

  An arbitrary discrete 2-form $F\in K(4)$  can be written  as follows
\begin{equation}\label{4.14}
F=\sum_k\sum_{i<j} F_k^{ij}\varepsilon_{ij}^k,
 \end{equation}
 where \ $ F_k^{ij}\in gl(2,\Bbb{C})$, \  $ \varepsilon_{ij}^k$ \ is the
  2-dimensional  basis element of \ $K(4)$ \ and \ $1\leq i,j\leq4$, \ $k=(k_1,k_2,k_3,k_4)$, \ $k_i\in\Bbb Z$.
Let $F$ is given by
\begin{equation}\label{4.15}
F=d^cA+A\cup A.
\end{equation}

Combining (\ref{4.7}) and (\ref{4.15}) and using (\ref{3.12}), (\ref{3.13}) and
(\ref{3.15}), we obtain
 \begin{equation}\label{4.16}
 F_k^{ij}=\Delta_iA_k^j-\Delta_jA_k^i+A_k^iA_{\tau_ik}^j-
 A_k^jA_{\tau_jk}^i,
 \end{equation}
 where $\Delta_iA_k^j=A_{\tau_ik}^j-A_k^j$ and $\tau_ik$ is given by (\ref{3.7}).

 Let us define a discrete analog of the exterior covariant
differentiation operator (\ref{2.10}) as follows
\begin{equation}\label{4.17}
d_A^c\Omega=d^c\Omega+A\cup\Omega+(-1)^{p+1}\Omega\cup A,
\end{equation}
 where $\Omega$ is an arbitrary $p$-form of $K(4)$ looking like (\ref{3.11}).
 Then a  discrete analog of Equations (\ref{2.13}) can be written as
\begin{equation}\label{4.18}
 d_A^c F=0, \quad d_A^c\ast\tilde\iota F=0,
\end{equation}
where $\tilde\iota$ is given by (\ref{3.18}).
 It is easy to check that the combinatorial Bianchi identity:
 \begin{equation}\label{4.19}
 d^cF+A\cup F-F\cup A=0
\end{equation}
holds for the discrete curvature form (\ref{4.15}) (cf.
(\ref{2.13})).

\begin{rem}In the continual case the curvature form $F$
 (\ref{2.9}) takes values in the algebra $su(2)$ for any $su(2)$-valued connection form $A$.
 Unfortunately, this is not true in the discrete case because, generally speaking, the components
 $A_k^iA_{\tau_ik}^j- A_k^jA_{\tau_jk}^i$ of the form $A\cup A$ (see (\ref{4.16})) do not belong to $su(2)$.
\end{rem}

 To define an $su(2)$-valued discrete analog of the curvature 2-form  we use the quaternionic form of $A$ (\ref{4.8}) and put in (\ref{4.15}). Then the discrete curvature form $F$ is given by
 \begin{equation}\label{4.20}
F=\IM\{d^cf\cup e+(f\cup e)\cup(f\cup e)\}.
\end{equation}
It should be noted that in the discrete case  calculation of the
imaginary part of $f\cup e$ and computing its curvature do not
commute.
\begin{prop}
 If $A = \IM(x^{-1}\cup d^cx)$, where $x$ is given by (\ref{4.10}), then   $F=0$.
\end{prop}
\begin{proof}
Using (\ref{4.5}) and putting $f=x^{-1}$ in (\ref{4.20}) we get
\begin{align*}
F&=\IM(d^c(x^{-1}\cup d^cx)+(x^{-1}\cup d^cx)\cup (x^{-1}\cup d^cx)\\
&=\IM(d^cx^{-1}\cup d^cx-d^cx^{-1}\cup x\cup x^{-1}\cup d^cx).
\end{align*}
According to (\ref{4.4}) the form $x\cup x^{-1}$ has unit components. Hence
\begin{equation*}
d^cx^{-1}\cup x\cup x^{-1}\cup d^cx=d^cx^{-1}\cup d^cx.
\end{equation*}
\end{proof}
We now write down the components of (\ref{4.14}) using quaternions.
Putting (\ref{4.9}) in (\ref{4.16}) we find that
\begin{align*}
F_k^{12}&=(\Delta_1f_k^1-\Delta_2f_k^2-f_k^3f_{\tau_1k}^3-f_k^4f_{\tau_1k}^4-f_k^3f_{\tau_2k}^3-f_k^4f_{\tau_2k}^4)\mathbf{i}\\
&+(\Delta_1f_k^4-\Delta_2f_k^3+f_k^2f_{\tau_1k}^3+f_k^4f_{\tau_1k}^1+f_k^1f_{\tau_2k}^4+f_k^3f_{\tau_2k}^2)\mathbf{j}\\
&+(-\Delta_1f_k^3-\Delta_2f_k^4+f_k^2f_{\tau_1k}^4-f_k^3f_{\tau_1k}^1-f_k^1f_{\tau_2k}^3+f_k^4f_{\tau_2k}^2)\mathbf{k}\\
&-f_k^2f_{\tau_1k}^1-f_k^3f_{\tau_1k}^4+f_k^4f_{\tau_1k}^3+f_k^1f_{\tau_2k}^2+f_k^4f_{\tau_2k}^3-f_k^3f_{\tau_2k}^4,
\end{align*}
\begin{align*}
F_k^{13}&=(-\Delta_1f_k^4-\Delta_3f_k^2+f_k^3f_{\tau_1k}^2-f_k^4f_{\tau_1k}^1-f_k^1f_{\tau_3k}^4+f_k^2f_{\tau_3k}^3)\mathbf{i}\\
&+(\Delta_1f_k^1-\Delta_3f_k^3-f_k^2f_{\tau_1k}^2-f_k^4f_{\tau_1k}^4-f_k^4f_{\tau_3k}^4-f_k^2f_{\tau_3k}^2)\mathbf{j}\\
&+(\Delta_1f_k^2-\Delta_3f_k^4+f_k^2f_{\tau_1k}^1+f_k^3f_{\tau_1k}^4+f_k^4f_{\tau_3k}^3+f_k^1f_{\tau_3k}^2)\mathbf{k}\\
&+f_k^2f_{\tau_1k}^4-f_k^3f_{\tau_1k}^1-f_k^4f_{\tau_1k}^2-f_k^4f_{\tau_3k}^2+f_k^1f_{\tau_3k}^3+f_k^2f_{\tau_3k}^4,
\end{align*}
\begin{align*}
F_k^{14}&=(\Delta_1f_k^3-\Delta_4f_k^2+f_k^3f_{\tau_1k}^1+f_k^4f_{\tau_1k}^2+f_k^2f_{\tau_4k}^4+f_k^1f_{\tau_4k}^3)\mathbf{i}\\
&+(-\Delta_1f_k^2-\Delta_4f_k^3-f_k^2f_{\tau_1k}^1+f_k^4f_{\tau_1k}^3+f_k^3f_{\tau_4k}^4-f_k^1f_{\tau_4k}^2)\mathbf{j}\\
&+(\Delta_1f_k^1-\Delta_4f_k^4-f_k^2f_{\tau_1k}^2-f_k^3f_{\tau_1k}^3-f_k^3f_{\tau_4k}^3-f_k^2f_{\tau_4k}^2)\mathbf{k}\\
&-f_k^2f_{\tau_1k}^3+f_k^3f_{\tau_1k}^2-f_k^4f_{\tau_1k}^1+f_k^3f_{\tau_4k}^2-f_k^2f_{\tau_4k}^3+f_k^1f_{\tau_4k}^4,
\end{align*}
\begin{align*}
F_k^{23}&=(-\Delta_2f_k^4-\Delta_3f_k^1+f_k^4f_{\tau_2k}^2+f_k^3f_{\tau_2k}^1+f_k^1f_{\tau_3k}^3+f_k^2f_{\tau_3k}^4)\mathbf{i}\\
&+(\Delta_2f_k^1-\Delta_3f_k^4-f_k^1f_{\tau_2k}^2+f_k^3f_{\tau_2k}^4+f_k^4f_{\tau_3k}^3-f_k^2f_{\tau_3k}^1)\mathbf{j}\\
&+(\Delta_2f_k^2+\Delta_3f_k^3+f_k^1f_{\tau_2k}^1+f_k^4f_{\tau_2k}^4+f_k^4f_{\tau_3k}^4+f_k^1f_{\tau_3k}^1)\mathbf{k}\\
&+f_k^1f_{\tau_2k}^4-f_k^4f_{\tau_2k}^1+f_k^3f_{\tau_2k}^2-f_k^4f_{\tau_3k}^1+f_k^1f_{\tau_3k}^4-f_k^2f_{\tau_3k}^3,
\end{align*}
\begin{align*}
F_k^{24}&=(\Delta_2f_k^3-\Delta_4f_k^1+f_k^4f_{\tau_2k}^1-f_k^3f_{\tau_2k}^2-f_k^2f_{\tau_4k}^3+f_k^1f_{\tau_4k}^4)\mathbf{i}\\
&+(-\Delta_2f_k^2-\Delta_4f_k^4-f_k^1f_{\tau_2k}^1-f_k^3f_{\tau_2k}^3-f_k^3f_{\tau_4k}^3-f_k^1f_{\tau_4k}^1)\mathbf{j}\\
&+(\Delta_2f_k^1+\Delta_4f_k^3-f_k^1f_{\tau_2k}^2-f_k^4f_{\tau_2k}^3-f_k^3f_{\tau_4k}^4-f_k^2f_{\tau_4k}^1)\mathbf{k}\\
&-f_k^1f_{\tau_2k}^3+f_k^4f_{\tau_2k}^2+f_k^3f_{\tau_2k}^1+f_k^3f_{\tau_4k}^1-f_k^2f_{\tau_4k}^4-f_k^1f_{\tau_4k}^3,
\end{align*}
\begin{align*}
F_k^{34}&=(\Delta_3f_k^3+\Delta_4f_k^4+f_k^1f_{\tau_3k}^1+f_k^2f_{\tau_3k}^2+f_k^2f_{\tau_4k}^2+f_k^1f_{\tau_4k}^1)\mathbf{i}\\
&+(-\Delta_3f_k^2-\Delta_4f_k^1+f_k^4f_{\tau_3k}^1+f_k^2f_{\tau_3k}^3+f_k^3f_{\tau_4k}^2+f_k^1f_{\tau_4k}^4)\mathbf{j}\\
&+(\Delta_3f_k^1-\Delta_4f_k^2+f_k^4f_{\tau_3k}^2-f_k^1f_{\tau_3k}^3-f_k^3f_{\tau_4k}^1+f_k^2f_{\tau_4k}^4)\mathbf{k}\\
&+f_k^4f_{\tau_3k}^3+f_k^1f_{\tau_3k}^2-f_k^2f_{\tau_3k}^1-f_k^3f_{\tau_4k}^4-f_k^2f_{\tau_4k}^1+f_k^1f_{\tau_4k}^2.
\end{align*}
To obtain the components of (\ref{4.20}) we must take the imaginary part of these equations.

\begin{prop} The discrete curvature 2-form $F$ (\ref{4.15}) is $su(2)$-valued if and only if
\begin{align*}
-f_k^2f_{\tau_1k}^1-f_k^3f_{\tau_1k}^4+f_k^4f_{\tau_1k}^3+f_k^1f_{\tau_2k}^2+f_k^4f_{\tau_2k}^3-f_k^3f_{\tau_2k}^4=0,\\
f_k^2f_{\tau_1k}^4-f_k^3f_{\tau_1k}^1-f_k^4f_{\tau_1k}^2-f_k^4f_{\tau_3k}^2+f_k^1f_{\tau_3k}^3+f_k^2f_{\tau_3k}^4=0,\\
-f_k^2f_{\tau_1k}^3+f_k^3f_{\tau_1k}^2-f_k^4f_{\tau_1k}^1+f_k^3f_{\tau_4k}^2-f_k^2f_{\tau_4k}^3+f_k^1f_{\tau_4k}^4=0,\\
f_k^1f_{\tau_2k}^4-f_k^4f_{\tau_2k}^1+f_k^3f_{\tau_2k}^2-f_k^4f_{\tau_3k}^1+f_k^1f_{\tau_3k}^4-f_k^2f_{\tau_3k}^3=0,\\
-f_k^1f_{\tau_2k}^3+f_k^4f_{\tau_2k}^2+f_k^3f_{\tau_2k}^1+f_k^3f_{\tau_4k}^1-f_k^2f_{\tau_4k}^4-f_k^1f_{\tau_4k}^3=0,\\
f_k^4f_{\tau_3k}^3+f_k^1f_{\tau_3k}^2-f_k^2f_{\tau_3k}^1-f_k^3f_{\tau_4k}^4-f_k^2f_{\tau_4k}^1+f_k^1f_{\tau_4k}^2=0.
\end{align*}
\end{prop}
\begin{proof}
From the above it follows immediately.
\end{proof}
\begin{prop} Let $e$ is given by (\ref{4.6}). Then the 2-form $e\cup\bar e$ is self-dual, i.e.,
\begin{equation}\label{4.21}
 e\cup\bar e=\ast\tilde\iota (e\cup\bar e),
\end{equation}
and  $\bar e\cup e$ is anti-self-dual, i.e.,
\begin{equation}\label{4.22}
 \bar e\cup e=-\ast\tilde\iota(\bar e\cup e).
\end{equation}
\end{prop}
\begin{proof}
Denote
\begin{equation*}
e_i=\sum_ke_i^k, \qquad  \varepsilon_{ij}=\sum_k\varepsilon_{ij}^k.
\end{equation*}
Recall that $e_i^k$  and $\varepsilon_{ij}^k$ are the 1-dimensional
and 2-dimensional basic elements of $K(4)$ (see also (\ref{3.4}) and
(\ref{3.5})). From this by (\ref{3.15}) we obtain $e_i\cup
e_j=\varepsilon_{ij}$ and $e_j\cup e_i=-\varepsilon_{ij}$ for all
$i<j$. Then we have
\begin{align*}
e\cup\bar e&=(e_1+e_2\mathbf{i}+e_3\mathbf{j}+e_4\mathbf{k})\cup(e_1-e_2\mathbf{i}-e_3\mathbf{j}-e_4\mathbf{k})\\
&=-2\{(e_1\cup e_2+e_3\cup e_4)\mathbf{i}+(e_1\cup e_3-e_2\cup e_4)\mathbf{j}+(e_1\cup e_4+e_2\cup e_3)\mathbf{k}\}\\
&=-2\{(\varepsilon_{12}+\varepsilon_{34})\mathbf{i}+(\varepsilon_{13}-\varepsilon_{24})\mathbf{j}+
(\varepsilon_{14}+\varepsilon_{23})\mathbf{k}\}.
\end{align*}
Using (\ref{3.17}) and (\ref{3.19}) we get
\begin{equation*}
\ast\tilde\iota (e\cup\bar e)=-2\tilde\iota\{(\tilde{\varepsilon}_{34}+\tilde{\varepsilon}_{12})\mathbf{i}+(-\tilde{\varepsilon}_{24}+\tilde{\varepsilon}_{13})\mathbf{j}+
(\tilde{\varepsilon}_{23}+\tilde{\varepsilon}_{14})\mathbf{k}\}=e\cup\bar e.
\end{equation*}
In the same way we obtain (\ref{4.22}).
\end{proof}
\begin{cor} For any quaternionic 0-form $f$ the form $f\cup e\cup\bar e$ is self-dual and $f\cup \bar e\cup e$ is anti-self-dual.
\end{cor}
\begin{proof}
This follows immediately from (\ref{3.20}).
\end{proof}
Discrete  self-dual and anti-self-dual equations (discrete analogs  of Equations (\ref{2.13})) are defined by
\begin{equation}\label{4.23}
 F=\tilde\iota\ast F, \qquad F=-\tilde\iota\ast F,
\end{equation}
where $F$ is the discrete curvature form (\ref{4.4}). Using
(\ref{4.5}), by the definitions of $\tilde\iota$ and $\ast$, the
first equation (self-dual) of (\ref{5.1}) can be rewritten as
follows
\begin{equation}\label{4.24}
 F_k^{12}=F_k^{34}, \qquad F_k^{13}=-F_k^{24}, \qquad
 F_k^{14}=F_k^{23}.
\end{equation}
By analogue with the continual case solutions of (\ref{4.23}) (or
(\ref{4.24})) are called  instantons and anti-instantons respectively.

\section{Discrete instanton and anti-instanton}
In further  analogy with the continual case consider the discrete
$SU(2)$-connection $A$. Let $A$ be the quaternionic 1-form  (\ref{4.8}), where the components of  $f$ are given by
\begin{equation}\label{5.1}
f_k=\frac{\bar{\kappa}}{1+|\kappa|^2},
\end{equation}
where
$\kappa=k_{1}+k_{2}\mathbf{i}+k_{3}\mathbf{j}+k_{4}\mathbf{k}$, \ $k_i\in\Bbb Z.$
Putting the last in (\ref{4.9}) we obtain
\begin{align}\label{5.2}
A_k^1=\frac{-k_{2}\mathbf{i}-k_{3}\mathbf{j}-k_{4}\mathbf{k}}{1+|\kappa|^{2}},  \qquad
A_k^2=\frac{k_{1}\mathbf{i}-k_{4}\mathbf{j}+k_{3}\mathbf{k}}{1+|\kappa|^{2}},\notag \\
A_k^3=\frac{k_{4}\mathbf{i}+k_{1}\mathbf{j}-k_{2}\mathbf{k}}{1+|\kappa|^{2}}, \qquad
A_k^4=\frac{-k_{3}\mathbf{i}+k_{2}\mathbf{j}+k_{1}\mathbf{k}}{1+|\kappa|^{2}}.
\end{align}
It is convenient to denote
\begin{equation}\label{5.3}
M_i=\frac{1}{(1+|\kappa|^2)(1+|\tau_i\kappa|^2)}, \qquad i=1,2,3,4.
\end{equation}
Recall that the shift operator $\tau_i$ is given by (\ref{3.7}).
 Substituting (\ref{5.2}) in (\ref{4.16}) and using (\ref{5.3}) we find that
\begin{align*}
F_k^{12}&=\{M_1(1+k_2^2-k_1^2-k_1)+M_2(1+k_1^2-k_2^2-k_2)\}\mathbf{i}\\
&+\{M_1(k_4k_1+k_2k_3)-M_2(k_3k_2+k_4k_1)\}\mathbf{j}\\
&+\{M_1(k_2k_4-k_1k_3)+M_2(k_1k_3-k_2k_4)\}\mathbf{k}\\
&+M_1(k_1k_2+k_2)-M_2(k_1k_2+k_1),
\end{align*}
\begin{align*}
F_k^{13}&=\{M_1(k_2k_3-k_1k_4)+M_3(k_1k_4-k_2k_3)\}\mathbf{i}\\
&+\{M_1(1+k_3^2-k_1^2-k_1)+M_3(1+k_1^2-k_3^2-k_3)\}\mathbf{j}\\
&+\{M_1(k_1k_2+k_3k_4)-M_3(k_3k_4+k_1k_2)\}\mathbf{k}\\
&+M_1(k_1k_3+k_3)-M_3(k_1k_3+k_1),
\end{align*}
\begin{align*}
F_k^{14}&=\{M_1(k_1k_3+k_2k_4)-M_4(k_2k_4+k_1k_3)\}\mathbf{i}\\
&+\{M_1(k_3k_4-k_1k_2)+M_4(k_1k_2-k_3k_4)\}\mathbf{j}\\
&+\{M_1(1+k_4^2-k_1^2-k_1)+M_4(1+k_1^2-k_4^2-k_4)\}\mathbf{k}\\
&+M_1(k_1k_4+k_4)-M_4(k_1k_4+k_1),
\end{align*}
\begin{align*}
F_k^{23}&=\{-M_2(k_2k_4+k_1k_3)+M_3(k_1k_3+k_2k_4)\}\mathbf{i}\\
&+\{M_2(k_3k_4-k_1k_2)+M_3(k_1k_2-k_3k_4)\}\mathbf{j}\\
&-\{M_2(1+k_3^2-k_2^2-k_2)+M_3(1+k_2^2-k_3^2-k_3)\}\mathbf{k}\\
&+M_2(k_2k_3+k_3)-M_3(k_2k_3+k_2),
\end{align*}
\begin{align*}
F_k^{24}&=\{M_2(k_2k_3-k_4k_1)+M_4(k_1k_4-k_2k_3)\}\mathbf{i}\\
&+\{M_2(1+k_4^2-k_2^2-k_2)+M_4(1+k_2^2-k_4^2-k_4)\}\mathbf{j}\\
&-\{M_2(k_1k_2+k_3k_4)-M_4(k_3k_4+k_1k_2)\}\mathbf{k}\\
&+M_2(k_2k_4+k_4)-M_4(k_2k_4+k_2),
\end{align*}
\begin{align*}
F_k^{34}&=-\{M_3(1+k_4^2-k_3^2-k_3)+M_4(1+k_3^2-k_4^2-k_4)\}\mathbf{i}\\
&+\{M_3(-k_2k_3-k_1k_4)+M_4(k_1k_4+k_2k_3)\}\mathbf{j}\\
&+\{M_3(k_2k_4-k_1k_3)+M_4(k_1k_3-k_2k_4)\}\mathbf{k}\\
&+M_3(k_3k_4+k_4)-M_4(k_3k_4+k_3).
\end{align*}
\begin{prop} The 2-form $F$ with components $F_k^{ij}$   above   is $su(2)$-valued if and only if
\begin{equation}\label{5.4}
k_1=k_2=k_3=k_4.
\end{equation}
\end{prop}
\begin{proof} From Proposition~4.4 $F$ is $su(2)$-valued if and only if
\begin{equation*}
M_i(k_ik_j+k_j)-M_j(k_ik_j+k_i)=0
\end{equation*}
for any $k_i\in\Bbb{Z}$, \ $i,j=1,2,3,4$ and $i<j$.  It follows
immediately (\ref{5.4}).
\end{proof}

Thus, the $su(2)$-valued discrete curvature 2-form $F$ can be
written in the quaternionic form as follows
\begin{equation}\label{5.5}
F=\sum_{k, \
k_i=\mu}M_\mu(2-2\mu)\{(\varepsilon_{12}^k-\varepsilon_{34}^k)\mathbf{i}
+(\varepsilon_{13}^k+\varepsilon_{24}^k)\mathbf{j}+(\varepsilon_{14}^k-\varepsilon_{23}^k)\mathbf{k}\}.
\end{equation}
 From (\ref{5.2}) here we have
 $M_\mu=\frac{1}{2(1+4\mu^2)(1+\mu+2\mu^2)}$. Since $k_i=\mu$,
 in (\ref{5.5}) we can write $\varepsilon_{ij}^{\mu}$  instead of
 $\varepsilon_{ij}^k$.

If we consider the 0-form
\begin{equation}\label{5.6}
\omega=\sum_{\mu}M_\mu(1-\mu)x^\mu, \qquad \mu\in\Bbb{Z}
\end{equation}
and use the following relation (see the proof of Proposition~4.5)
\begin{equation*}
\bar e\cup
e=2\{(\varepsilon_{12}-\varepsilon_{34})\mathbf{i}+(\varepsilon_{13}+\varepsilon_{24})\mathbf{j}+
(\varepsilon_{14}-\varepsilon_{23})\mathbf{k}\},
\end{equation*}
 then
$F$ can be written as
\begin{equation}\label{5.7}
F=\omega\cup \bar{e}\cup e.
\end{equation}
In view of Corollary~4.6 $F$ is anti-self-dual, i.e., $F=-\tilde\iota\ast
F$. Thus under condition (\ref{5.4})  $A$ with components (\ref{5.1})
describes an anti-instanton.

In the same manner we can see that the following quaternionic 1-form
\begin{equation}\label{5.8}
A=\IM(f\cup\bar{e}),
\end{equation}
where $f$ has the components
\begin{equation}\label{5.9}
f_k=\frac{\kappa}{1+|\kappa|^2},
\end{equation}
leads to an instanton solution of (\ref{4.24}).
Indeed, substituting (\ref{5.8}) and (\ref{5.9}) in (\ref{4.16}) we now obtain

\begin{align*}
F_k^{12}&=\{-M_1(1+k_2^2-k_1^2-k_1)-M_2(1+k_1^2-k_2^2-k_2)\}\mathbf{i}\\
&+\{M_1(k_4k_1-k_2k_3)+M_2(k_3k_2-k_4k_1)\}\mathbf{j}\\
&+\{M_1(-k_2k_4-k_1k_3)+M_2(k_1k_3+k_2k_4)\}\mathbf{k}\\
&+M_1(k_1k_2+k_2)-M_2(k_1k_2+k_1),
\end{align*}
\begin{align*}
F_k^{13}&=\{M_1(-k_2k_3-k_1k_4)+M_3(k_1k_4+k_2k_3)\}\mathbf{i}\\
&-\{M_1(1+k_3^2-k_1^2-k_1)+M_3(1+k_1^2-k_3^2-k_3)\}\mathbf{j}\\
&+\{M_1(k_1k_2-k_3k_4)+M_3(k_3k_4-k_1k_2)\}\mathbf{k}\\
&+M_1(k_1k_3+k_3)-M_3(k_1k_3+k_1),
\end{align*}
\begin{align*}
F_k^{14}&=\{M_1(k_1k_3-k_2k_4)+M_4(k_2k_4-k_1k_3)\}\mathbf{i}\\
&+\{M_1(-k_3k_4-k_1k_2)+M_4(k_1k_2+k_3k_4)\}\mathbf{j}\\
&-\{M_1(1+k_4^2-k_1^2-k_1)+M_4(1+k_1^2-k_4^2-k_4)\}\mathbf{k}\\
&+M_1(k_1k_4+k_4)-M_4(k_1k_4+k_1),
\end{align*}
\begin{align*}
F_k^{23}&=\{M_2(-k_2k_4+k_1k_3)+M_3(-k_1k_3+k_2k_4)\}\mathbf{i}\\
&+\{M_2(k_3k_4+k_1k_2)-M_3(k_1k_2+k_3k_4)\}\mathbf{j}\\
&-\{M_2(1+k_3^2-k_2^2-k_2)+M_3(1+k_2^2-k_3^2-k_3)\}\mathbf{k}\\
&+M_2(k_2k_3+k_3)-M_3(k_2k_3+k_2),
\end{align*}
\begin{align*}
F_k^{24}&=\{M_2(k_2k_3+k_4k_1)-M_4(k_1k_4+k_2k_3)\}\mathbf{i}\\
&+\{M_2(1+k_4^2-k_2^2-k_2)+M_4(1+k_2^2-k_4^2-k_4)\}\mathbf{j}\\
&+\{M_2(k_1k_2-k_3k_4)+M_4(k_3k_4-k_1k_2)\}\mathbf{k}\\
&+M_2(k_2k_4+k_4)-M_4(k_2k_4+k_2),
\end{align*}
\begin{align*}
F_k^{34}&=-\{M_3(1+k_4^2-k_3^2-k_3)+M_4(1+k_3^2-k_4^2-k_4)\}\mathbf{i}\\
&+\{M_3(-k_2k_3+k_1k_4)+M_4(-k_1k_4+k_2k_3)\}\mathbf{j}\\
&+\{M_3(k_2k_4+k_1k_3)-M_4(k_1k_3+k_2k_4)\}\mathbf{k}\\
&+M_3(k_3k_4+k_4)-M_4(k_3k_4+k_3).
\end{align*}
Again,  under condition (\ref{5.4}) we can write $F$ as
\begin{equation*}
F=\sum_{\mu}M_\mu(2\mu-2)\{(\varepsilon_{12}^\mu+\varepsilon_{34}^\mu)\mathbf{i}
+(\varepsilon_{13}^\mu-\varepsilon_{24}^\mu)\mathbf{j}+(\varepsilon_{14}^\mu+\varepsilon_{23}^\mu)\mathbf{k}\},
\end{equation*}
where  $\mu\in\Bbb{Z}$.  Therefore
\begin{equation}\label{5.10}
F=\omega\cup e\cup\bar{e},
\end{equation}
where  $\omega$ is given by (\ref{5.6}). Thus the discrete curvature form (\ref{5.10}) is self-dual and we can say that (\ref{5.8}) describes an instanton.

Now to complete the analogy with the continual case we describe more
precisely how the anti-instanton given by (\ref{5.1}) behaves as
$|\kappa|\rightarrow\infty$. It is clear that $f_k$ is
asymptotically $\frac{\bar{\kappa}}{|\kappa|^2}=\kappa^{-1}$. Then
 \begin{equation}\label{5.11}
  A \sim \IM(x^{-1}\cup d^cx) \quad \mbox{as} \quad |\kappa|\rightarrow\infty.
\end{equation}
Here $x$ is given by (\ref{4.10}). By virtue of Proposition~4.3 the
discrete curvature  $F=0$ at infinity.
\begin{prop} The anti-instanton  (\ref{5.1}) has the same form at $\infty$ as it
has near 0.
\end{prop}
\begin{proof} Introduce the quaternionic 0-form
\begin{equation*}
y=\sum_ky_kx^k,  \qquad \mbox{where} \qquad  y_k=\frac{1}{\kappa}
\end{equation*}
 and remind $\kappa=k_{1}+k_{2}\mathbf{i}+k_{3}\mathbf{j}+k_{4}\mathbf{k}$.
 Clearly, $y=x^{-1}$. We first compute  $x\cup f\cup e\cup x^{-1}$, where  $f$ is given by (\ref{5.1}). To do this, take  (\ref{4.10}), (\ref{4.11}) and use the $\cup$-product definition. We have
 \begin{align*}
 x\cup f\cup e&=\Big(\sum_k\kappa x^k\Big)\cup \Big(\sum_k\frac{\bar{\kappa}}{1+|\kappa|^2}x^k\Big)\cup e\\
 &=\Big(\sum_k\frac{|\kappa|^2}{1+|\kappa|^2}x^k\Big)\cup e=e-\Big(\sum_k\frac{1}{1+|\kappa|^2}x^k\Big)\cup e\\
 &=d^cx-\Big(\sum_k\frac{1}{1+|\kappa|^2}x^k\Big)\cup d^cx.
 \end{align*}
 From this by (\ref{4.5})  we get
 \begin{equation}\label{5.12}
 x\cup f\cup e\cup x^{-1}=-x\cup d^cx^{-1}+\Big(\sum_k\frac{\kappa}{1+|\kappa|^2}x^k\Big)\cup d^cx^{-1}.
 \end{equation}
 Now gauge transform the form $f\cup e$ by the gauge transformation $g=x^{-1}$. We must take the imaginary part of (\ref{4.13}). This yields by (\ref{5.12})
 \begin{align*}
 \IM(g^{-1}\cup f\cup e\cup g+g^{-1}\cup d^cg)&=\IM\Big(\Big(\sum_k\frac{\kappa}{1+|\kappa|^2}x^k\Big)\cup d^cx^{-1}\Big)\\
 &=\IM\Big(\Big(\sum_k\frac{\bar{y}_k}{1+|y_k|^2}x^k\Big)\cup d^cy\Big).
 \end{align*}
 Hence the gauge transformed anti-instanton $A$  has  precisely the  form (\ref{5.11}) near $y=0$.
\end{proof}
 The same conclusion can be drawn for the instanton (\ref{5.8}).

In the continual theory Proposition~5.2 shows that the
anti-instanton (or instanton) extends to the 4-sphere $S^4$. This
follows from the fact that $S^4$ can be obtained from $\Bbb{R}^4$ by
adding the point at infinity, i.e., $S^4\simeq\Bbb{R}^4\cup
\{\infty\}$. To obtain the same result for our discrete model we
need to construct a suitable combinatorial analog of the 4-sphere.
It would be interesting to connect the above constructions with
discrete model of $S^4$ described in \cite{S3}. This connection must
be investigated and we hope to treat its further in future work.

\end{document}